\documentclass[letterpaper, 10 pt, conference]{ieeeconf}  

\IEEEoverridecommandlockouts                              

\title{\LARGE \bf
Learning Control Barrier Functions with High Relative Degree \\ for Safety-Critical Control 
}

\author{Chuanzheng Wang, Yinan Li, Yiming Meng, Stephen L. Smith, Jun Liu
\thanks{Chuanzheng Wang, Yinan Li, Yiming Meng and Jun Liu are with the Department of Applied Mathematics,        
University of Waterloo, Waterloo, Ontario, Canada, 
        {\tt\small \{cz.wang, yinan.li, yiming.meng, j.liu\}@uwaterloo.ca}}%
\thanks{Stephen L. Smith is with the Department of Electrical and Computer Engineering, 
University of Waterloo, Waterloo, Ontario, Canada, 
        {\tt\small stephen.smith@uwaterloo.ca}}%
}

\usepackage{amsmath,amssymb,amstext}
\usepackage{cite}
 \newtheorem{problem}{Problem}
 \usepackage{algorithm,algorithmic}
 \usepackage{graphicx}
\usepackage{subcaption}
 \usepackage{mwe}
 \renewcommand{\epsilon}{\varepsilon}
 \renewcommand{\theta}{\vartheta}
 \renewcommand{\kappa}{\varkappa}
 \renewcommand{\rho}{\varrho} 
 \renewcommand{\phi}{\varphi}

 \newcommand{\Real}{\mathbb{R}}

 \newcommand{\BlackBox}{\rule{1.5ex}{1.5ex}}
 
 \newtheorem{definition}{Definition}
 \newtheorem{remark}{Remark}
 \newtheorem{prop}{Proposition}
 
 \newtheorem{lma}{Lemma}
 \newtheorem{cor}{Corollary}

\usepackage{xcolor}

\DeclareMathOperator*{\argmin}{arg\,min}
\usepackage{booktabs}
\begin{document}

\maketitle
\thispagestyle{empty}
\pagestyle{empty}

\begin{abstract}
Control barrier functions have shown great success in addressing control problems with safety guarantees. These methods usually find the next safe control input by solving an online quadratic programming problem. However, model uncertainty is a big challenge in synthesizing controllers. This may lead to the generation of unsafe control actions, resulting in severe consequences. In this paper, we develop a learning framework to deal with system uncertainty. Our method mainly focuses on learning the dynamics of the control barrier function, especially for high relative degree with respect to a system. We show that for each order, the time derivative of the control barrier function can be separated into the time derivative of the nominal control barrier function and a remainder. This implies that we can use a neural network to learn the remainder so that we can approximate the dynamics of the real control barrier function. We show by simulation that our method can generate safe trajectories under parametric uncertainty using a differential drive robot model.

\end{abstract}

\section{INTRODUCTION}\label{sec:intro}
\subsection{Background and Literature Review}
In many applications, one must solve a control problem that requires not only achieving control objectives, but also providing control actions with guaranteed safety \cite{garcia2015comprehensive}. In practice, this is of great importance and it is necessary to incorporate safety criteria while designing controllers. For example, industrial robotics, medical robots as well as self-driving vehicles are all areas where safe controllers are critical. The notion of safety control was first proposed in \cite{lamport1977proving} in the form of correctness and was then formalized in \cite{alpern1985defining}, in which the authors stated that a safety property stipulates that some ``bad thing'' does not happen during execution. 

More recently, control barrier functions (CBFs) are widely used to deal with safety control \cite{ames2019control}. Barrier functions are Lyapunov-like functions which were initially used in optimization problems \cite{boyd2004convex}. CBFs are combined with control Lyapunov functions as constraints of quadratic programming (QP) problems in \cite{ames2014control} and the authors show that safety criteria can be converted into a linear constraint of the QP problem for control inputs. By solving the QP problems, we can find the next action so that safety is guaranteed during execution. It is shown in \cite{rauscher2016constrained} that finding safe control inputs by solving QP problems can be extended to an arbitrary number of constraints and any nominal control law. As a result, CBFs are widely used in safety control such as lane keeping \cite{ames2016control} and obstacle avoidance \cite{chen2017obstacle}. However, using CBFs in the QP problems means that the first order derivative of the CBFs should depend on the control input and as a result, this usually violates with many robot systems such as bipedal or car-like robots \cite{hsu2015control}. Consequently, CBFs are extended to handle position-based constraints for relative degree of two \cite{wu2015safety}. The authors in \cite{nguyen2016exponential} propose a way of designing exponential control barrier functions (ECBFs) using input-output linearization to handle CBFs with higher relative degree. Safe control actions are calculated for quadrotors using ECBF in \cite{wang2017safe}. Furthermore, a more general form of higher-order control barrier functions (HOCBFs) is introduced in \cite{xiao2019control}. 

In practice models used to design controllers are imperfect because of disturbance or parametric uncertainty. This uncertainty may lead to unsafe or even dangerous behavior, and thus it is of great importance that we synthesize controllers to handle model uncertainty. Learning-based approaches have shown great promise in controlling systems with uncertainty \cite{khansari2014learning}. Several methods using data-driven approaches have been utilized in this area. Training data is collected to learn the real dynamics for the design of more accurate controllers. In \cite{yaghoubi2020training}, the HOCBF under external disturbance is proposed and imitation learning is used to obtain a feedback controller. Gaussian process (GP) is used to approximate the model as in \cite{chowdhary2014bayesian}. The authors in \cite{cheng2019end} also use Gaussian process to estimate the model but focusing on the safety during the training process. A reinforcement learning (RL) based method to learn the model uncertainty compensation for input-output linearization control is introduced in \cite{westenbroek2019feedback} and a RL-based framework for policy improvement is proposed in \cite{choi2020reinforcement} as well. However, both methods do not rely on a nominal controller and using nominal controllers are more flexible because they can be replaced by any other reliable controllers in practice. Our work is mostly close to \cite{taylor2020learning}, in which the dynamics of the CBF of real model is learned based on the dynamics of the CBF for the nominal model. However, the main difference between our work and \cite{taylor2020learning} is that we focus on learning CBFs with higher relative degree with respect to more complex systems. Besides, we also provide sufficient conditions on controllers via CBFs with high relative degree for set invariance.
\subsection{Contribution}
In this paper, we propose a learning framework for CBFs with high relative degree. We consider a machine learning method to reduce model uncertainty using supervised regression. Safe trajectories are generated using the learned CBFs. It is shown that the dynamics of the real CBF can be learned based on the nominal CBF. As a result, the main contribution of our work is summarized as below:
\begin{itemize}
    \item [1] We propose a learning framework for CBFs with high relative degree for safety-critical control.
    \item [2] We provide sufficient conditions on controllers via CBFs with high relative degree for set invariance.
    \item [3] We show theoretically that for each order of time derivative,  the dynamics of the real CBF can be separated into two terms: the time derivative of the nominal CBF and a remainder that is independent of the control input.
    \item [4] We use supervised regression to learn the remainder so that the dynamics of the real CBF can be accurately approximated.
    \item [5] We validate our method in simulation using a differential drive model under system uncertainty with static obstacles, multiple obstacles and moving obstacles. 
\end{itemize}

\section{Preliminary and Problem Definition} \label{sec:definition}
\subsection{Model and Uncertainty}
Throughout the paper, we consider a SISO nonlinear control affine model
\begin{equation}\label{eq:sys}
\begin{split}
    \dot{x}&=f(x)+g(x)u,\\
    y&=h(x),
\end{split}
\end{equation}
such that $f:\mathbb{R}^n\to\mathbb{R}^n$ and $g:\mathbb{R}^n\to\mathbb{R}^n$ are locally Lipschitz, $x\in\mathbb{R}^n$ is the state and $u\in\mathbb{R}$ is the control and $h:\mathbb{R}^n\to\mathbb{R}$ is a $r^{\text{th}}$-order continuously differentiable function for some integer $r$. A solution of system (\ref{eq:sys}) from an initial condition $x_0\in\mathbb{R}^n$ is denoted by $x(t,x_0)$.

We also consider parametric uncertainty for the model, and as a result, we have a nominal model that estimates the dynamics of Eq~(\ref{eq:sys}) as
\begin{equation}\label{eq:nominal}
\begin{split}
    \hat{\dot{x}}&=\hat{f}(x)+\hat{g}(x)u,\\
    \hat{y}&=\hat{h}(x),
\end{split}
\end{equation}
where $\hat{f}:\mathbb{R}^n\to\mathbb{R}^n$ and $\hat{g}:\mathbb{R}^n\to\mathbb{R}^n$ are locally Lipschitz continuous and $\hat{h}:\mathbb{R}^n\to\mathbb{R}$ is an $r^{\text{th}}$-order continuously differentiable function as well.
\subsection{Control Barrier Function}
We consider a set $\mathcal{C}$ defined as a superlevel set of a continuously differentiable function $h:\mathbb{R}^n\to\mathbb{R}$ such that 
\begin{equation}\label{eq:c}
    \begin{split}
        \mathcal{C}&=\{x\in\mathbb{R}^n : h(x)\geq 0\},\\
        \partial{\mathcal{C}}&=\{x\in \mathbb{R}^n : h(x)=0\},\\
        \text{Int}(\mathcal{C})&=\{x\in \mathbb{R}^n : h(x)>0\}.
    \end{split}
\end{equation}
We refer $\mathcal{C}$ as the safe set and safety can be framed in the context of enforcing invariance of $\mathcal{C}$. Due to the local Lipschitz assumption of $f$ and $g$, for any initial condition $x_0$, there exists a maximum interval of existence $I(x_0)=[0,\tau_{\text{max}})$ such that $x(t, x_0)$ is the unique solution to (\ref{eq:sys}) on $I(x_0)$. 
As a result, we can define a set to be forward invariant as below.

\begin{definition}
Let $h:\mathbb{R}^n\to\mathbb{R}$ be a continuously differentiable function and  $\mathcal{C}\subset\mathbb{R}^n$ be a superlevel set of $h$ as defined in Eq~(\ref{eq:c}). The set $\mathcal{C}$ is forward invariant if for every $x_0\in\mathcal{C}$, $x(t)\in\mathcal{C}$ for all $t\in I(x_0)$, where $x(t)$ is the solution to Eq~(\ref{eq:sys}) with $x(0)=x_0$. The system (\ref{eq:sys}) is safe with respect to $\mathcal{C}$ if $\mathcal{C}$ is forward invariant.
\end{definition}

We note that an extended $\mathcal{K}_\infty$ function is a function $\alpha:\mathbb{R}\to\mathbb{R}$ that is strictly increasing and $\alpha(0)=0$. Based on this, we can define the control barrier function as follows.
\begin{definition}\label{def:cbf}
Let $h:\mathbb{R}^n\to\mathbb{R}$ be a continuously differentiable function and  $\mathcal{C}\subset \mathcal{D}\subset\mathbb{R}^n$ be a superlevel set of $h$ as defined in Eq~(\ref{eq:c}). Then $h$ is a control barrier function (CBF) if there exists an extended $\mathcal{K}_{\infty}$ function $\alpha$ such that for the control system Eq~(\ref{eq:sys}),
\begin{equation*}
    \sup_{u\in\mathbb{R}}{[L_fh(x)+L_gh(x)u]}\geq-\alpha({h(x)})
\end{equation*}
for all $x\in\mathcal{D}$, where $L_fh(x)=f\cdot\frac{\partial h}{\partial x}$ and $L_gh(x)=g\cdot\frac{\partial h}{\partial x}$.
\end{definition}

We can then consider the set consisting of all control values that render $\mathcal{C}$ to be safe\cite{ames2019control}:
\begin{equation*}
    K_{\text{cbf}}=\{u(x)\in\mathbb{R}:L_fh(x)+L_gh(x)u+\alpha({h(x)})\geq 0\}.
\end{equation*}
\subsection{Safety-Critical Control}
Suppose we are given a feedback controller $u=k(x)$ for the system (\ref{eq:sys}) and we wish to control the system while guaranteeing safety. It may be the case that sometimes the feedback controller $u=k(x)$ is not safe, i.e., there exists some $x$ such that $u(x)\notin K_{\text{cbf}}=\{u(x)\in\mathbb{R}:L_fh(x)+L_gh(x)u+\alpha(h(x))\geq 0\}$. We can use the following quadratic programming to find the safe control with minimum perturbation \cite{ames2016control}:
\begin{equation*}
    \begin{split}
        u(x)&=\argmin_{u\in\mathbb{R}}\frac{1}{2}||u-k(x)||^2\quad\quad\quad\text{(CBF-QP)}\\
        &\text{s.t.}\quad L_fh(x)+L_gh(x)u+\alpha({h(x)})\geq 0.
    \end{split}
\end{equation*}

\subsection{Relative Degree and Exponential Control Barrier Function}

The relative degree of a continuous differentiable function $h$ on a set with respect to a system as in Eq~(\ref{eq:sys}) is the number of times we need to differentiate $h$ along the dynamics of the system before the control input $u$ explicitly appears. The formal definition of relative degree is as below.
\begin{definition}
Given an $r^{\text{th}}$-order continuously differentiable function $h$, a set $D$ and a system as defined in Eq~(\ref{eq:sys}), we say $h$ has a relative degree of $r$ with respect to system Eq~(\ref{eq:sys}) on $D$ if $L_gL^{r-1}_fh(x)\neq 0$ and $L_gL_fh(x)=L_gL_f^2h(x)=\dots=L_gL_f^{r-2}h(x)=0$ for $x\in D$, where 
$L^r_fh(x)=L_fL^{r-1}_fh(x)$.
\end{definition}

\begin{remark}
 In this paper, we assume that $h$ has a well-defined relative degree of $r$ with respect to system Eq~(\ref{eq:sys}) on a domain $D$ of interest, similar to \cite{xu2018constrained}, where the author assumed $D=\Real^n$. 
 \end{remark}

The $r^{\text{th}}$-order time-derivative of $h(x)$ is
\begin{equation*}
    h^r(x)=L^r_fh(x)+L_gL^{r-1}_fh(x)u
\end{equation*}
and $h^r(x)$ is dependent on the control input $u$. The system is input-output linearizable if $L_gL_f^{r-1}h(x)$ is invertible. For a given control $\mu\in\mathbb{R}$, $u$ can be chosen such that $L_f^rh(x)+L_gL_f^{r-1}h(x)u=\mu$. The control input $u$ renders the input-output dynamics of the system linear. Defining a system with state
\begin{equation*}
\eta(x):=\begin{bmatrix}
h(x)\\
\dot{h}(x)\\
\vdots\\
h^{r-1}(x)\end{bmatrix}=\begin{bmatrix}
h(x)\\
L_fh(x)\\
\vdots\\
L_f^{r-1}h(x)
\end{bmatrix},
\end{equation*}
we can then construct a state-transformed linear system
\begin{equation}\label{sys:transfer}
    \begin{split}
        \dot{\eta}(x)&=F\eta(x)+G\mu,\\
        h(x)&=C\eta(x),
    \end{split}
\end{equation}
where
\begin{equation*}
    \begin{split}
        F&=\begin{bmatrix}
            0&1&0&\dots&0\\
            0&0&1&\dots&0\\
            \vdots&\vdots&\vdots&\ddots&\vdots\\
            0&0&0&\dots&1\\
            0&0&0&\dots&0\end{bmatrix},\quad G=\begin{bmatrix}
                                                    0\\
                                                    0\\
                                                    \vdots\\
                                                    0\\
                                                    1\end{bmatrix},\\
        C&=\begin{bmatrix}1&0&0&\dots&0\end{bmatrix}.
    \end{split}
\end{equation*}
The exponential control barrier function is defined below as in \cite{nguyen2016exponential}. 

\begin{definition}\label{def:ecbf}
Given a $r^{\text{th}}$-order continuously differentiable function $h:\mathbb{R}^n\to\mathbb{R}$ and a superlevel set $\mathcal{C}$ of $h$ as defined in Eq~(\ref{eq:c}), then $h$ is an exponential control barrier function (ECBF) if there exists a row vector $K=[k_0,k_1,\dots,k_{r-1}]$ such that 
\begin{equation*}
    \sup_{u\in\mathbb{R}}{[L_f^rh(x)+L_gL_f^{r-1}h(x)u]}\geq -K\eta(x).
\end{equation*}
for any $x\in\mathcal{C}$, where $K$ is chosen such that the transformed system Eq~(\ref{sys:transfer}) is stable.
\end{definition}

\begin{remark}
It is explained in \cite{nguyen2016exponential} that the ECBF with $r=1$ is the same as the CBF as in Definition~\ref{def:cbf}. The design of the ECBF, i.e, the selection of $k_0,k_1,\dots,k_{r-1}$ in $K$ is also explained in \cite{nguyen2016exponential} using state feedback control and pole placement.
\end{remark}

As a result, given an ECBF and a nominal controller $u=k(x)$, we can consider the following quadratic programming problem to enforce the condition in Definition \ref{def:ecbf} with minimum perturbation
\begin{equation*}
    \begin{split}
        u(x)&=\argmin_{u\in\mathbb{R}}\frac{1}{2}||u-k(x)||^2\quad\quad\quad\text{(ECBF-QP)}\\
        &\text{s.t.}\quad L_f^rh(x)+L_gL_f^{r-1}h(x)u\geq -K\eta(x).
    \end{split}
\end{equation*}

\subsection{High Order Control Barrier Function and Controlled Set Invariance}

Exponential control barrier functions (ECBF) can be seen as a special case of higher order control barrier functions (HOCBF) defined in \cite{xiao2019control}. In this section, we present some sufficient conditions on using HOCBF for enforcing set invariance. We first define a series of continuously differentiable function $b_0,b_j:\mathbb{R}^n\to\mathbb{R}$ for each $j=1,2,\dots,r$ and corresponding superlevel sets $\mathcal{C}_j$ as
\begin{equation}\label{eq:hocbf}
\begin{split}
    b_0(x)&=h(x),\\
    b_j(x)&=\dot{b}_{j-1}(x)+c_j\alpha_j(b_{j-1}(x)),\\
\end{split}
\end{equation}
and
\begin{equation}\label{eq:hoset}
\begin{split}
  \mathcal{C}_j&=\{x\in\mathbb{R}^n:b_{j-1}(x)\geq 0\},\\
\end{split}
\end{equation}
where $c_j>0$ are constants and $\alpha_j(\cdot)$ are differentiable extended class $\mathcal{K}$ functions. We further assume that the interiors of the sets $\mathcal{C}_i$ are given by
$$
\text{Int}(\mathcal{C}_i) = \{x\in\mathbb{R}^n:b_{j-1}(x)> 0\}.
$$

\begin{definition}\label{def:hocbf}
A continuously differentiable function $h$ is an $r^{\text{th}}$-order control barrier function (HOCBF) for system (\ref{eq:sys}), if there exists extended differentiable class $\mathcal{K}$ functions $\alpha_j(\cdot)$ for $j=1,2,...,r$, such that for $b_j(x)$ defined in Eq~(\ref{eq:hocbf}) with any arbitrary $c_j>0$ and the corresponding superlevel sets $\mathcal{C}_j$ defined as in Eq~(\ref{eq:hoset}), the following 
\begin{equation}\label{E: HOCBF}
    \sup\limits_{u\in\Real}[L_f^rh(x)+L_gL_f^{r-1}h(x)u+\mathcal{O}(h)]\geq -c_r\alpha_r(b_{r-1}(x))
\end{equation}
holds for all $x\in\bigcap_{j=1}^{r}\mathcal{C}_j$, where $\mathcal{O}(h)$ denotes the Lie derivatives of $h$ along $f$ with degree up to $r-1$.
\end{definition}
\begin{remark}
Note that $\mathcal{C}_1$ is uniquely defined, whereas $\mathcal{C}_2, \mathcal{C}_3,...,\mathcal{C}_r$ is defined based on the choice of $c_1,c_2,...,c_{r-1}$.
\end{remark}
\begin{prop}\label{prop: set-invariance}
Consider $r^{\text{th}}$-order HOCBF $h:\Real^n\rightarrow \Real$ with the associated $\alpha_j$ and sets $\mathcal{C}_i$ for $j\in\{1,2,...,r\}$. Suppose that $h$ has relative degree $r$ with respect to system (\ref{eq:sys}) on a set $D$ containing 
$\bigcap_{j=1}^{r} \mathcal{C}_j$. Then any Lipschitz continuous controller $u(x)$ that satisfies 
\begin{equation}\label{eq:controller}
    L_f^rh(x)+L_gL_f^{r-1}h(x)u(x)+\mathcal{O}(h)\geq -c_r\alpha_r(b_{r-1}(x))
\end{equation}
for all $x\in \bigcap_{j=1}^{r} \text{Int}(\mathcal{C}_j)$ renders the set $\bigcap_{j=1}^{r} \text{Int}(\mathcal{C}_j)$ forward invariant. Furthermore, given any functions $\alpha_j$, $j\in\{1,2,...,r\}$, and any compact initial set $X_0\subset \text{Int}(\mathcal{C}_1)$, there exist appropriate choices of $c_j>0$ such that $X_0\subset \bigcap_{j=1}^{r} \text{Int}(\mathcal{C}_j)$.
\end{prop}

Before proceeding to the proof, we introduce technical tools to show how the invariance conditions is effective for first order barrier functions. We first cite a lemma from \cite{glotfelter2017nonsmooth}, which can be proved based on Lemma 4.4 in \cite{khalil2002nonlinear} and well-known comparison techniques \cite{lakshmikantham1969differential}. 

\begin{lma}\label{lma: khalil}\cite{glotfelter2017nonsmooth}
Let $z:[t_0,t_f)\rightarrow \Real$ be a continuously differentiable function satisfying the differential inequality 
\begin{equation}\label{E: inequality}
    \dot{z}(t)\geq -\alpha(z(t)),\;\forall t\in[t_0,t_f),
\end{equation}
 where $\alpha:\,\Real\rightarrow\Real$ is a locally Lipschitz extended class $\mathcal{K}$ function. Then there exists a class $\mathcal{KL}$ function $\beta:\,[0,\infty)\times [0,\infty)\rightarrow [0,\infty)$ (only depending on $\alpha$) such that 
 $$
 z(t) \ge \beta(z(t_0),t-t_0),\quad \forall t\in [t_0,t_f). 
 $$
\end{lma}

\begin{cor}\label{cor: cor}
Given a continuously differentiable function $h:\Real^n\rightarrow \Real$ and 
dynamics on $\Real^n$
\begin{equation}
    \dot{x}=f(x)
\end{equation}
such that $f:\Real^n\rightarrow\Real$ is locally Lipschitz. Let $\mathcal{C}=\{x:\; h(x)\geq 0\}$, and $\text{Int}(\mathcal{C}):=\{x:\; h(x)> 0\}$. If the Lie derivative of $h$ along the trajectories of $x$ satisfies 
\begin{equation}\label{E: invariant}
    \dot{h}(x)\geq -\alpha(h(x)), \;\forall x\in\mathcal{C}
\end{equation}
where $\alpha$ is a  locally  Lipschitz  extended  class $\mathcal{K}$ function, then the set $\text{Int}(\mathcal{C})$ is forward invariant.
\end{cor}
\begin{proof}
If $\text{Int}(\mathcal{C})=\emptyset$, then it is invariant. Otherwise, we apply Lemma \ref{lma: khalil}, it follows that if $x(t_0)\in \text{Int}(\mathcal{C})$, then we have $h(x(t))>0$ for all $t\in[t_0,t_f)$,  where $[t_0,t_f)$ is the maximal interval of existence for $x(t)$ starting from $x(t_0)$. 
\end{proof}

\begin{remark}\label{rem:alpha}
Note that the result cannot be extended to the invariance of the set $\mathcal{C}$, despite that it is widely stated so in the literature. A simple counterexample is when $\text{Int}(\mathcal{C})=\emptyset$, we can define $h(x)=-x^2$ and therefore $\mathcal{C}=\{0\}$. Then for $\dot{x}=c\neq 0$, even though we have a satisfaction of \eqref{E: invariant} on $\mathcal{C}=\{0\}$, it is not invariant under the flow. 

Now assume $\text{Int}(\mathcal{C})\neq\emptyset$, we also need to necessarily assume the locally Lipschitz continuity of $\alpha$. As for a counter example, let $\dot{x}=-1$ and $h(x)=\frac{2\sqrt{2}}{3\sqrt{3}}x^{3/2}$ for $x\geq 0$. Then the point $0$ loses asymptotic behavior and $h(x)$ will reach $0$ within finite time for any $x_0>0$. 

\end{remark}

\textit{Proof of Proposition \ref{prop: set-invariance}}: 

By the choice of controller $u(x)$ in (\ref{eq:controller}), we have
\begin{equation}\label{eq:estimate}
    b_r(x)=\dot{b}_{r-1}(x)+c_r\alpha_r(b_{r-1}(x))\geq 0
\end{equation}
for all $x\in \bigcap_{j=1}^{r} \mathcal{C}_j$. Suppose $x_0\in  \bigcap_{j=1}^{r} \text{Int}(\mathcal{C}_j)$. Then there exists a small time $\tau>0$ such that the solution to (\ref{eq:sys}) under the controller $u(x)$ is defined on $[0,\tau]$ and $x(t)\in \bigcap_{j=1}^{r} \text{Int}(\mathcal{C}_j)$ for all $t\in [0,\tau]$. The differentiability of $\alpha_j$ implies its local Lipschitz continuity. By (\ref{eq:estimate}) and Lemma \ref{lma: khalil}, we have $b_{r-1}(x(t))>0$ for all $t\in [0,\tau]$. By the same argument, we can show that $b_{j}(x(t))>0$ for all $t\in [0,\tau]$ and all $j=0,1,\cdots,r-1$. To conclude that $x(t)\in \bigcap_{j=1}^{r} \text{Int}(\mathcal{C}_j)$ for all $t$ in the maximal interval of existence of $x(t)$, we can use the fact that the $\mathcal{KL}$ lower bound given by Lemma \ref{lma: khalil} only depends on $\alpha_j$'s.

As for any given $\alpha_j$ and $X_0\subset \text{Int}(\mathcal{C}_1)$, since $L_f^jh$, $\mathcal{O}_{j-1}(h)$, $\alpha_j$ for $j\in\{1,...,r\}$ and $L_gL_f^{r-1}hu$ are all continuous functions, we can recursively define $c_j>\max(-\frac{L_f^jh(x_0)+\mathcal{O}_{j-1}h(x_0)}{\alpha_j(b_{j-1}(x_0))}, \delta_j)$ for arbitrary $\delta_j>0$ from $j=1$ to $j=r-1$. Similarly, we choose $c_r>\max(-\frac{L_f^jh(x_0)+L_gL_f^{r-1}h(x_0)u(x_0)+\mathcal{O}_{r-1}h(x_0)}{\alpha_r(b_{r-1}(x_0))}, \delta_r)$. The above choice of $c_j$ for $j\in\{1,...,r\}$ guarantees
$b_j(x_0)=\dot{b}_{j-1}(x_0)+c_j\alpha_j(b_{j-1}(x_0))> 0$, or equivalently
$X_0\subset  \bigcap_{j=1}^{r}\text{Int}(\mathcal{C}_j) $.  \hfill\BlackBox

\begin{remark}
Sufficient conditions for enforcing set invariance using ECBF or HOCBF can be found in \cite{nguyen2016exponential} and \cite{xiao2019control}, respectively (see also \cite{ames2019control}). The first part of Proposition \ref{prop: set-invariance} recaptures the results in \cite{nguyen2016exponential,xiao2019control,ames2019control}, but we spell out the importance of the local Lipschitz condition on $\alpha_j$'s and the fact that the set $\bigcap_{j=1}^{r}\mathcal{C}_j$ itself may not be controlled invariant under the well-known (zeroing) CBF condition (see Remark \ref{rem:alpha} above) even for the case $r=1$ without further assumptions. The second part of Proposition \ref{prop: set-invariance} is in case that by a given HOCBF $h$, the initial point $x_0\notin  \bigcap_{j=1}^{r}\mathcal{C}_j $. However, one can always rescale the existing $\alpha_j$ with proper choices of $c_j$ to provide invariance conditions, 
such that controllers adjusted to the conditions will lead the trajectories starting from any compact initial set $X_0 \subset \text{Int}(\mathcal{C}_1)$ invariant within $\mathcal{C}_1$. 
\end{remark}

\subsection{Problem Formulation}
The objective of this paper is to control a nonlinear system (\ref{eq:sys}) with unknown parameters to reach a given target set while ensuring safety, i.e., staying inside a safe set. We assume that the nominal model (\ref{eq:nominal}) is known and there exists a nominal feedback controller such that the closed-loop system can safely reach the target set. Then the problem is formally formulated as below.
\begin{problem}
Given system as in Eq~(\ref{eq:sys}), a goal region $\mathcal{X}_{\text{goal}}\subset\mathbb{R}^n$, a safe set $\mathcal{X}_{\text{safe}}\subset\mathbb{R}^n$, a nominal controller $k(x)$, and an initial state $x_{\text{init}}$, design a feedback controller $u=\tilde{k}(x)$, where $x\in\mathcal{X}_{\text{safe}}$ and $\tilde{k}:\mathcal{X}_{\text{safe}}\to\mathbb{R}$, such that the solution of the closed-loop system satisfies that $x(T, x_{\text{init}})\in\mathcal{X}_{\text{goal}}$ for some $T>0$ and $x(t, x_{\text{init}})\in\mathcal{X}_{\text{safe}}$ for all $t\geq 0$.
\end{problem}

\section{Model Uncertainty and Learning Framework}\label{sec:learning}
In this section, we discuss how we deal with model uncertainty and learn the real model. As defined in Section~\ref{sec:definition}, we consider a real model as in Eq~(\ref{eq:sys}), where $f$ and $g$ are not known precisely in practice and a nominal model as in Eq~(\ref{eq:nominal}) that estimates the true dynamics of the system is available. Then we can rewrite the real model using parametric uncertainty as 
\begin{equation}\label{eq:real}
    \dot{x}=\hat{{f}}(x)+\hat{{g}}(x)u+b(x)+A(x)u,
\end{equation}
where $b(x)=f(x)-\hat{f}(x)$ and $A(x)=g(x)-\hat{g}(x)$.

\begin{prop}
Given a nominal model and a real model as in Eq~(\ref{eq:nominal}) and Eq~(\ref{eq:real}), respectively, and the corresponding control barrier functions $\hat{h}$ and $h$ of the same relative degree $r$ with respect to the nominal model, real model and uncertainty on a set $D$, we have
\begin{equation}\label{eq:thm}
    h^m(x)=\hat{h}^m(x)+\Delta_m(x),\quad x\in D,
\end{equation}
for $m=1,2,\dots,r-1$, where $\Delta_m(x)$ is a remainder term that is independent of the control input $u$.
\end{prop}
\begin{proof}
We use mathematical induction to prove the result.

For $m=1$, we have 
\begin{equation*}
\begin{split}
\dot{h}(x)&=\frac{\partial h}{\partial x}\cdot(\hat{f}(x)+\hat{g}(x)u+b(x)+A(x)u)\\
&=L_{\hat{f}}h(x)+L_{\hat{g}}h(x)u+L_bh(x)+L_Ah(x)u.\\
\end{split}
\end{equation*}
Since $h$ is with high relative degree $r>1$ with respect to the nominal model and uncertainty, we have $L_{\hat{g}}h(x)=L_Ah(x)=0$, as a result, 
\begin{equation*}
\begin{split}
\dot{h}(x)&=L_{\hat{f}}h(x)+L_bh(x)\\
&=\hat{\dot{h}}(x)+L_bh(x)\\
&=\hat{\dot{h}}(x)+\Delta_1(x),
\end{split}
\end{equation*}
where $\Delta_1(x)=L_bh(x)$. We can see that for $m=1$, Eq~(\ref{eq:thm}) holds and $\Delta_1(x)$ is independent of $u$. Now assume that Eq~(\ref{eq:thm}) holds for $m=k$ and $\Delta_m(x)$ is independent of $u$, then for $m=k+1$, we have
\begin{equation*}
    \begin{split}
        &h^{k+1}(x)\\
        &\quad=\frac{(\partial{\hat{h}^k(x)}+\Delta_k(x))}{\partial{x}}\cdot(\hat{f}(x)+\hat{g}(x)u+b(x)+A(x)u)\\
        &\quad=L_{\hat{f}}^{k+1}h(x)+L_{\hat{g}}L_{\hat{f}}^kh(x)u+L_{b}L_{\hat{f}}^kh(x)+L_AL^k_{\hat{f}}h(x)u\\
        &\qquad\qquad+\frac{\partial{\Delta}_k(x)}{\partial{x}}\cdot(\hat{f}(x)+\hat{g}(x)u+b(x)+A(x)u)).\\
    \end{split}
\end{equation*}
Since the relative degree of $h$ with respect to the real model, nominal model and uncertainty are all $r$, for $m=k+1<r-1$, we have  $L_{\hat{g}}L_{\hat{f}}^kh(x)u=L_AL^k_{\hat{f}}h(x)u=0$, $L_{\hat{f}}^{k+1}h(x)=\hat{h}^{k+1}$ and $L_{\hat{g}}(\frac{\partial \Delta_k(x)}{\partial x})=L_{A}(\frac{\partial \Delta_k(x)}{\partial x})=0$. As a result, 
\begin{equation*}
    h^{k+1}(x)=\hat{h}^{k+1}(x)+\Delta_{k+1}(x)
\end{equation*}
such that 
\begin{equation*}
\begin{split}
\Delta_{k+1}(x)&=L_{b}L_{\hat{f}}^kh(x)\\
&+\frac{\partial{\Delta}_k(x)}{\partial{x}}\cdot(\hat{f}(x)+\hat{g}(x)u+b(x)+A(x)u)\\
&=L_{b}L_{\hat{f}}^kh(x)+\frac{\partial{\Delta}_k(x)}{\partial{x}}\cdot(\hat{f}(x)+b(x))\\
&=L_{b}L_{\hat{f}}^kh(x)+L_{\hat{f}}(\frac{\partial{\Delta}_k(x)}{\partial{x}})+L_b(\frac{\partial{\Delta}_k(x)}{\partial{x}}).
\end{split}
\end{equation*}
This means that for $m=k+1$, the equation $h^{k+1}(x)=\hat{h}^{k+1}(x)+\Delta_{k+1}(x)$ also holds and $\Delta_{k+1}(x)$ is independent of control input $u$. 
\end{proof}

The above proposition shows that for $m=1,2,\cdots, r-1$, we can always separate the time derivative of the CBF for the real system into the time derivative of the CBF for the nominal system and a remainder. As a result, for $m=r$:
\begin{equation*}
\begin{split}
h^{r}(x)&=\frac{\partial (L^{r-1}_{\hat{f}}h(x)+\Delta_{r-1}(x))}{\partial x}\cdot(\hat{f}(x)+\hat{g}(x)u\\
&+b(x)+A(x)u)\\
&=L^r_{\hat{f}}h(x)+L_{\hat{g}}L^{r-1}_{\hat{f}}h(x)+L_bL^{r-1}_{\hat{f}}h(x)\\
&+L_AL^{r-1}_{\hat{f}}h(x)u+\frac{\partial \Delta_{r-1}(x)}{\partial x}\cdot(\hat{f}(x)+b(x))\\
&+\frac{\partial \Delta_{r-1}}{\partial x}\cdot(\hat{g}(x)+A(x))u\\
&=\hat{h}^r+\Delta_r+\Sigma_ru,
\end{split}
\end{equation*}
where $\Delta_r(x)=\frac{\partial \Delta_{r-1}(x)}{\partial x}\cdot(\hat{f}(x)+b(x))+L_bL^{r-1}_{\hat{f}}h(x)$ and $\Sigma_r(x)=L_AL^{r-1}_{\hat{f}}h(x)+\frac{\partial \Delta_{r-1}(x)}{\partial x}\cdot(\hat{g}(x)+A(x))$.
According to the above conclusion, we know that the higher order time derivative of the real CBF $h^r$ can be separated into the higher order time derivative of the nominal CBF $\hat{h}^r$ and a remainder $\Delta_r+\Sigma_ru$. This implies that we can use neural networks to approximate $\Delta_r(x)$ and $\Sigma_r(x)$ via supervised regression. We can sample initial states and let the system evolve according to the given nominal controller. At each time step, we can store transition information into a buffer $\mathcal{B}=\{(x_i,u_i),h_i^r\}_i^N$, where $N$ is the length of the buffer. The term $h_i^r$ is calculated using numerical differentiation and this is the true value of $r^{\text{th}}$-order time derivative of CBF. Then we can construct an estimator to learn this true value using
\begin{equation*}
    \hat{\dot{E}}(x)=\hat{\dot{h}}^r+\Delta(x)+\Sigma(x)u.
\end{equation*}
Specifying a loss function $\mathcal{L}$ using minimum square error (MSE), the regression task is to find the estimator such that the loss function $\frac{1}{N}\sum_{i=1}^{N}\mathcal{L}(\hat{\dot{E}}(x),h^r(x))$ is minimized.
Meanwhile, a very important property of learning process is that the data has to be independently and identically distributed (i.i.d). Since the data generated along the trajectories violate this assumption, we use a buffer to store memory along trajectories as in \cite{mnih2013playing}. We first sample an initial point within working space and roll out according to the nominal controller. The control input executed during the transition is calculated by solving the quadratic programming problem
\begin{equation}\label{eq:QP}
    \begin{split}
        u(x)&=\argmin_{u\in\mathbb{R}}\frac{1}{2}||u-k(x)||^2\quad\quad\quad\text{(ECBF-QP)}\\
        &\text{s.t.}\quad \hat{\dot{E}}(x)\geq -K\eta(x),
    \end{split}
\end{equation}
as in \cite{ames2016control} but using $\hat{\dot{E}}(x)$ as the estimation of $h^r(x)$.
This quadratic programming problem helps to find a safe control that is nearest to the nominal control $k(x)$. The estimator is also improved along the sampling trajectories and is updated at each time step. At each time step, we sample data from buffer $\mathcal{B}$ and update neural networks such that the loss function is minimized. The algorithm of learning CBF with high relative degree is shown in Algorithm~\ref{alg:alg}. The algorithm will finally provide an estimator that is accurate enough to mimic the dynamics of the $r^{\text{th}}$-order time derivative of CBF for the real model and safe trajectories can be generated using the learned CBF. 
\begin{algorithm}[ht]
	\caption{Learning algorithm for CBFs with high relative degree}
	\label{alg:alg}
	\begin{algorithmic}[1]
	\REQUIRE A working space, a safe set, a nominal CBF $\hat{h}$, Dataset $\mathcal{B}$, nominal control policy $k(x)$, maximum step $n$ in each trajectory, initial neural network, number of trajectory sampled $\mathcal{N}$, batch size M, loss function $\mathcal{L}$.
	\STATE Initialize neural network and buffer $\mathcal{B}$
	\FOR{i in $\mathcal{N}$}
	\STATE Sample an initial point $x_0$
	\FOR{j in $1,2,\cdots,n$}
	\STATE Calculate control $u_j$ by solving QP problem in Eq~(\ref{eq:QP}) \
	\STATE Get $x_{j+1}$ from $x_j$ and $u_j$\
	\STATE $\mathcal{B}\leftarrow ((x_j,u_j), h^r_j)$
	\STATE Sample batch from $\mathcal{B}$
	\STATE Update neural network by minimizing the loss function $\mathcal{L}$
	\ENDFOR
	\ENDFOR
	\end{algorithmic}
\end{algorithm}

\section{Simulation Result}\label{sec:simulation}
In this section, we test our algorithm using a differential drive model as in \cite{lavalle2006planning}:
\begin{equation*}
\begin{split}
    \dot{x}&=r\frac{u_l+u_r}{2}\cos{\theta},\\
    \dot{y}&=r\frac{u_l+u_r}{2}\sin{\theta},\\
    \dot{\theta}&=\frac{r}{L}(u_l-u_r),\\
\end{split}
\end{equation*}
where $x$ and $y$ are the planar positions of the center of the vehicle, $\theta$ is its orientation, $r$ is the radius of the wheel, $L$ is the distance between two wheels and $u_l$, $u_r$ are angular velocity of left and right wheels, respectively. By substituting $u=\frac{u_l+u_r}{2}$ and $\omega=u_l-u_r$, we can get the following model
\begin{equation*}
\begin{split}
    \dot{x}&=ru\cos{\theta},\\
    \dot{y}&=ru\sin{\theta},\\
    \dot{\theta}&=\frac{r}{L}\omega,\\
\end{split}
\end{equation*}
where $\omega$ is the control input of the system.

\subsection{Experiment 1}
In the first experiment, we test our algorithm for single static obstacle avoidance. The working space is $[-3,3]\times[-3,3]\times[-\pi, \pi]$. The center of the obstacle is at the origin $(0,0)$ with radius $r_O=1.5$. The parametric uncertainty of the system comes from inaccurate measurement of the parameters $r$, $L$ and $u$. The safety requirement of the system is encoded as avoiding the obstacle successfully. This is expressed mathematically using a CBF
\begin{equation*}
    h(x,y,\theta)=x^2+y^2-r_O^2.
\end{equation*}
The CBF has a relative degree $r=2$ with respect to the system as there is no orientation $\theta$ in it. The corresponding ECBF is
\begin{align}
    2(ru)^2+&2\omega(y\cos{\theta}-x\sin{\theta})\frac{r^3u^2}{L}+k_1(x^2+y^2-r_O^2)\notag\\
    &+k_2(2ru)(x\cos{\theta}+y\sin{\theta})\geq 0.\label{eq:ecbf}
\end{align}
 The nominal policy is calculated using TRPO \cite{schulman2015trust} with 2 millions training steps in the working space without any obstacles. We solve quadratic programming problems using ECBF as in Eq~(\ref{eq:ecbf}). We test our algorithm under uncertainty in $r$, $L$ and $u$ separately. The parameters for the nominal model are $r=0.1$, $L=0.1$ and $u=1$ while the real system has parameters $r=0.07$, $L=0.13$ and $u=0.7$. The parameters of the first experiment is presented in TABLE~\ref{Tab:1}. In each case, we use a neural network with 2 hidden layers and 200 nodes in each layer. We sample 40 trajectories to train each network and the simulation results are shown in Figure~\ref{fig:Exp_1}. We can see that the trajectories calculated using the nominal CBFs are not safe for the real model. But after training the neural networks, the real CBFs are well learned to provide safe trajectories under parametric uncertainty.

\begin{table}[htbp]
\centering
\begin{tabular}{lcccccl}
\toprule
& r & L & u & $k_1$ & $k_2$& \\ 
\midrule
Nominal model & 0.1 & 0.1 & 1 & 1 & 6&\\ 
Uncertainty in $r$  & 0.07 & 0.1 & 1 & 1 & 6&\\
Uncertainty in $L$  & 0.1 & 0.13 & 1 & 1 & 6&\\ 
Uncertainty in $u$  & 0.1 & 0.1 & 0.7 & 1 & 6&\\ 
\bottomrule
\end{tabular}
\caption{Parameters in simulation for Experiment 1}
\label{Tab:1}
\end{table}

\begin{figure}[htbp]
	\centering
	\begin{subfigure}[h]{0.48\linewidth}
		\includegraphics[width=1\linewidth]{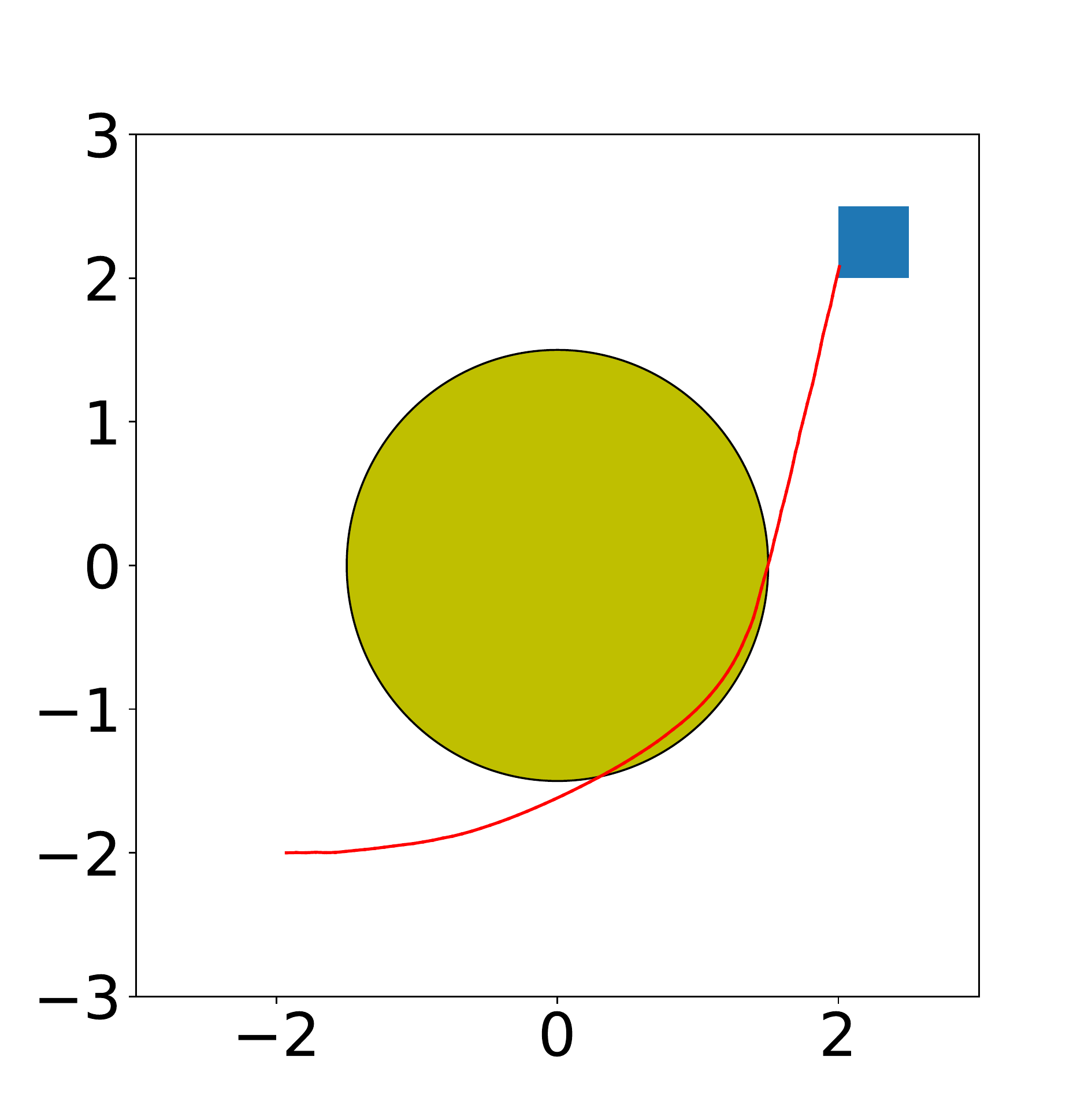}
	\caption{}
	\label{fig:2a}
	\end{subfigure}%
	\begin{subfigure}[h]{0.48\linewidth}
		\includegraphics[width=1\linewidth]{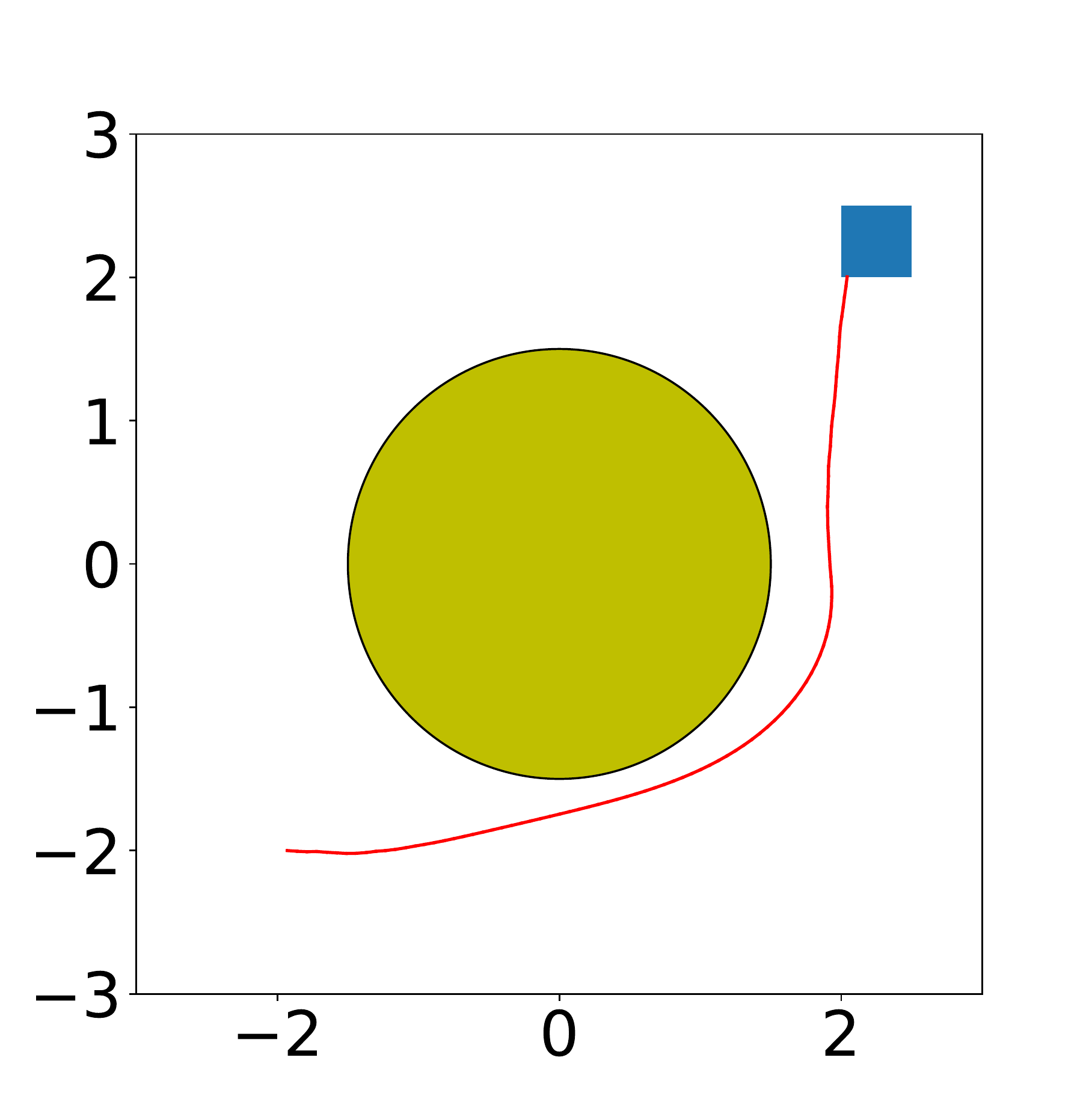}
	\caption{}
	\label{fig:2b}
	\end{subfigure}
	
	\begin{subfigure}[h]{0.48\linewidth}
	\includegraphics[width=1\linewidth]{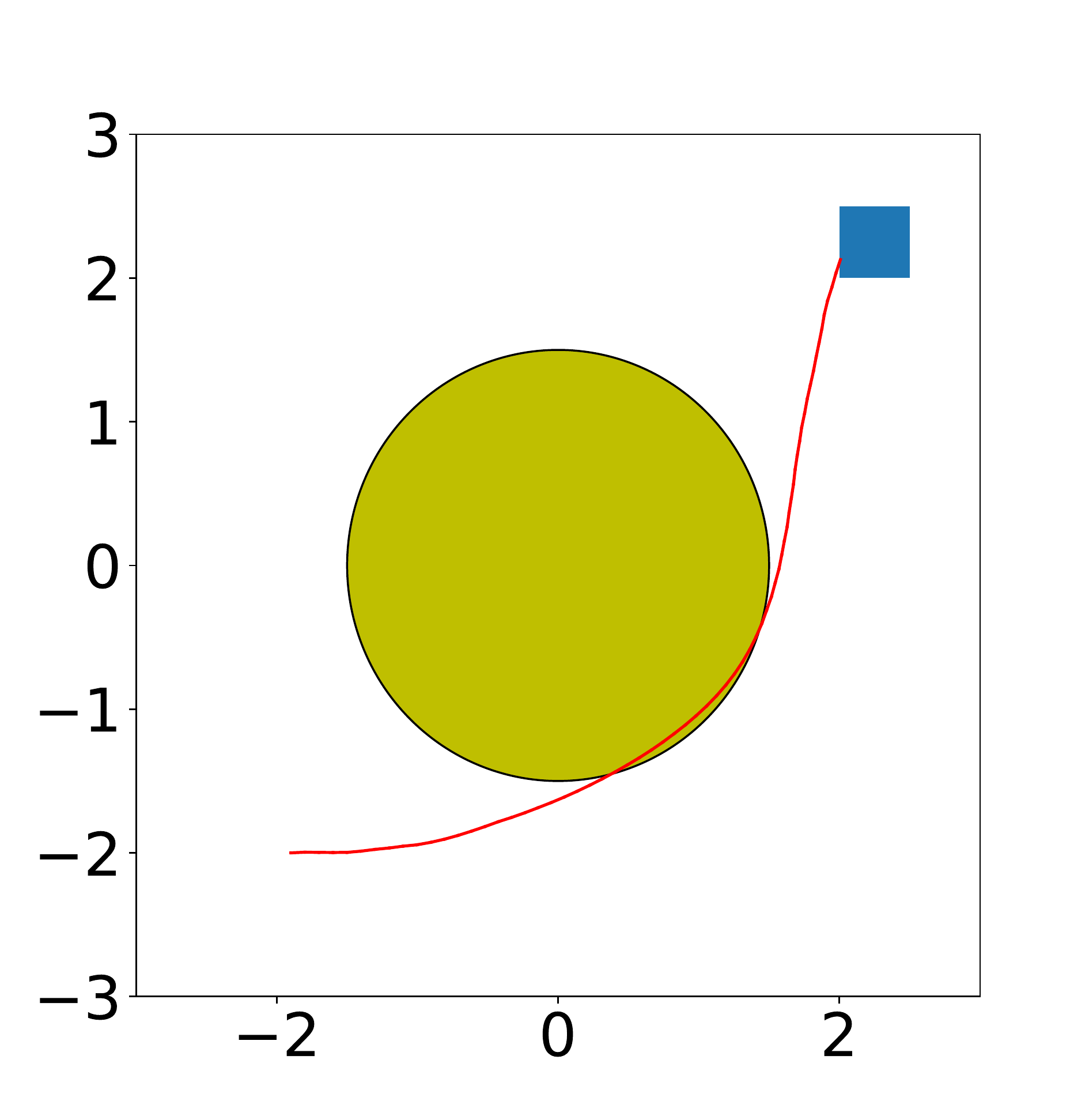}
	\caption{}
	\label{fig:3a}
	\end{subfigure}%
	\begin{subfigure}[h]{0.48\linewidth}
		\includegraphics[width=1\linewidth]{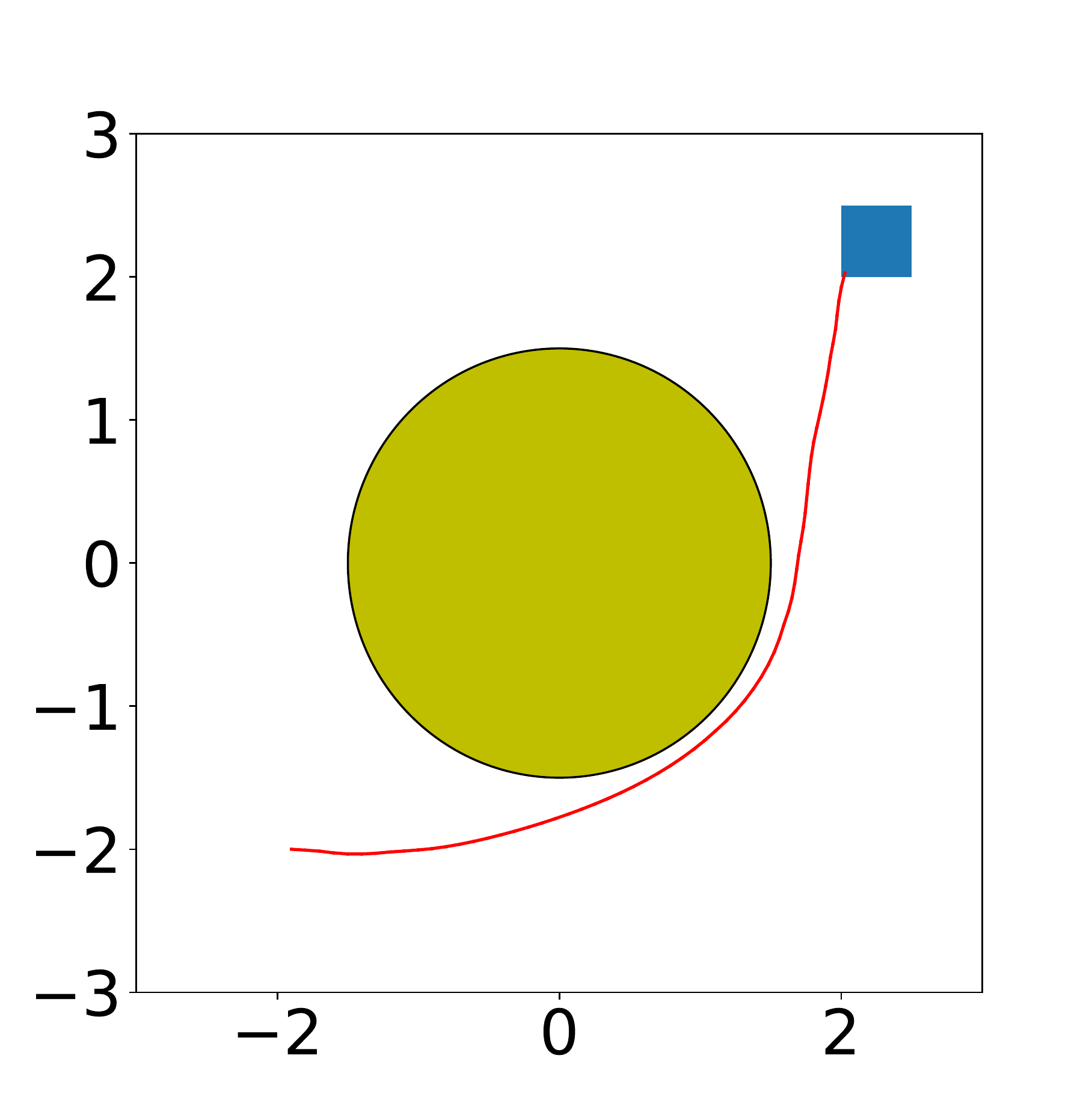}
	\caption{}
	\label{fig:3b}
	\end{subfigure}
	
	\begin{subfigure}[h]{0.48\linewidth}
	\includegraphics[width=1\linewidth]{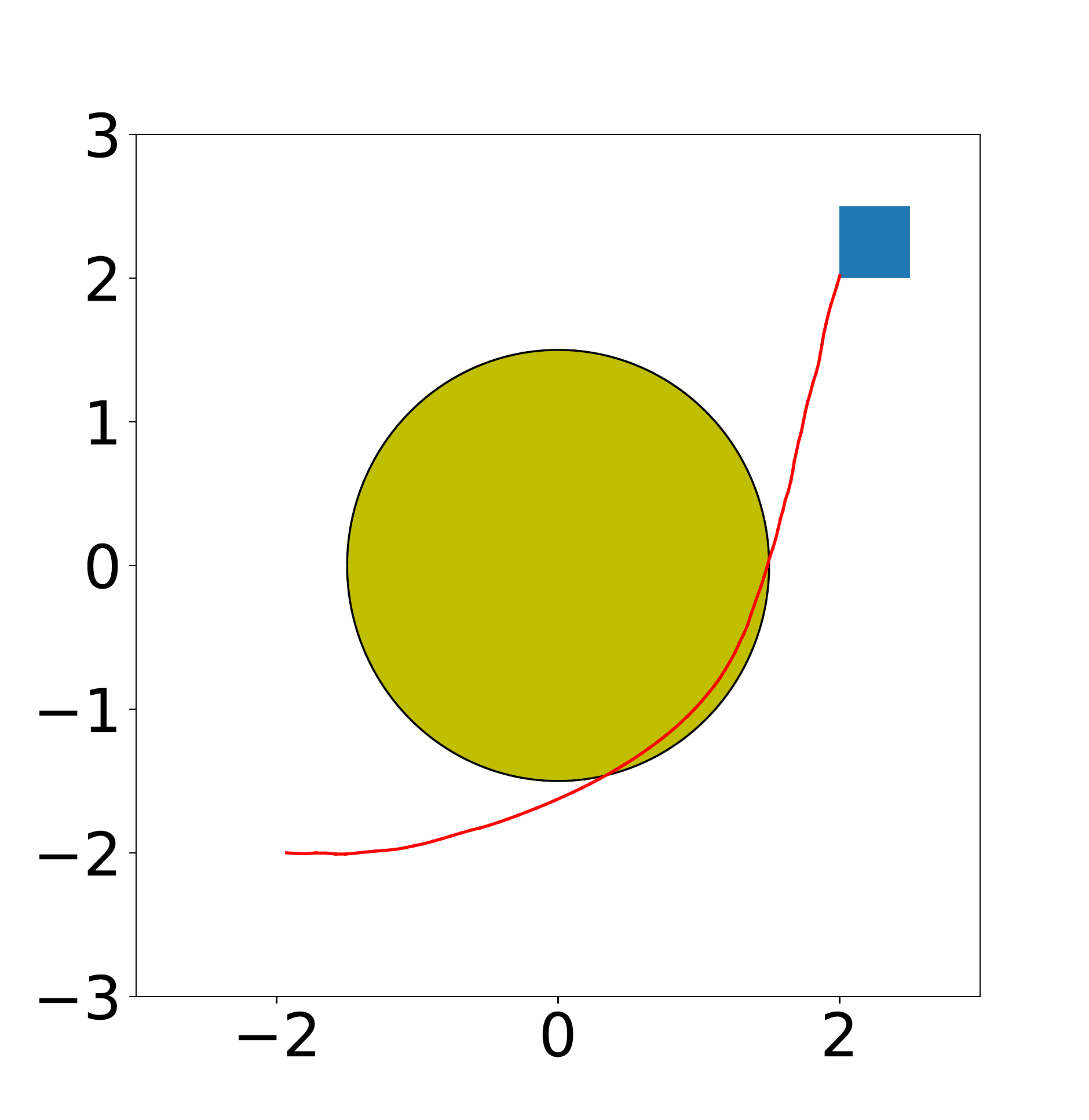}
	\caption{}
	\label{fig:4a}
	\end{subfigure}%
	\begin{subfigure}[h]{0.48\linewidth}
		\includegraphics[width=1\linewidth]{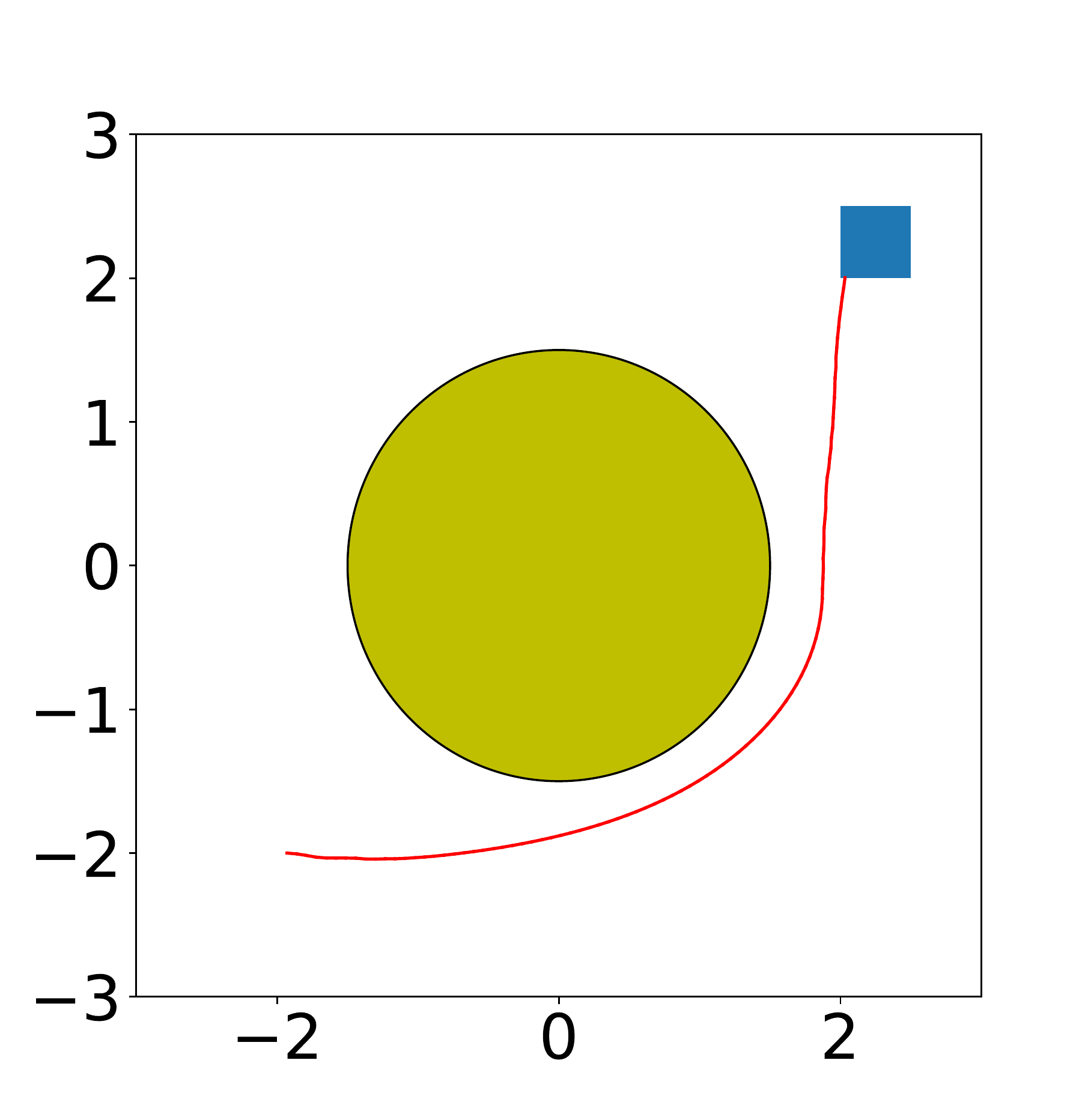}
	\caption{}
	\label{fig:4b}
	\end{subfigure}
	\caption{Simulation results for Experiment 1: The working space is $[-3,3]\times[-3,3]\times[-\pi, \pi]$. The yellow circles are the obstacles centered at $(0,0)$ with radius $r_O=1.5$. The blue squares are the goal regions and red curves are trajectories. (a), (c) and (e): The trajectories using the nominal CBF for real systems with uncertainty in $r$, $L$ and $u$, respectively. (b), (d), (f): The trajectories using the learned CBF for real systems with uncertainty in $r$, $L$ and $u$, respectively}
	\label{fig:Exp_1}
\end{figure}

We also test the safe rate for the trajectories with 50 initial points between using the nominal CBF and the learned CBF. Since the uncertainty will most likely make trajectories that are close to the obstacle unsafe, we only sample initial points from the green areas as in Fig~\ref{fig:5a} and Fig~\ref{fig:5b}. The result are shown in TABLE~\ref{Tab:2}. We can see that all the trajectories are safe by using the learned CBF while for the nominal CBF, the safe rate is only $28\%$.

\begin{figure}[htbp]
	\centering
	\begin{subfigure}[h]{0.5\linewidth}
		\includegraphics[width=1\linewidth]{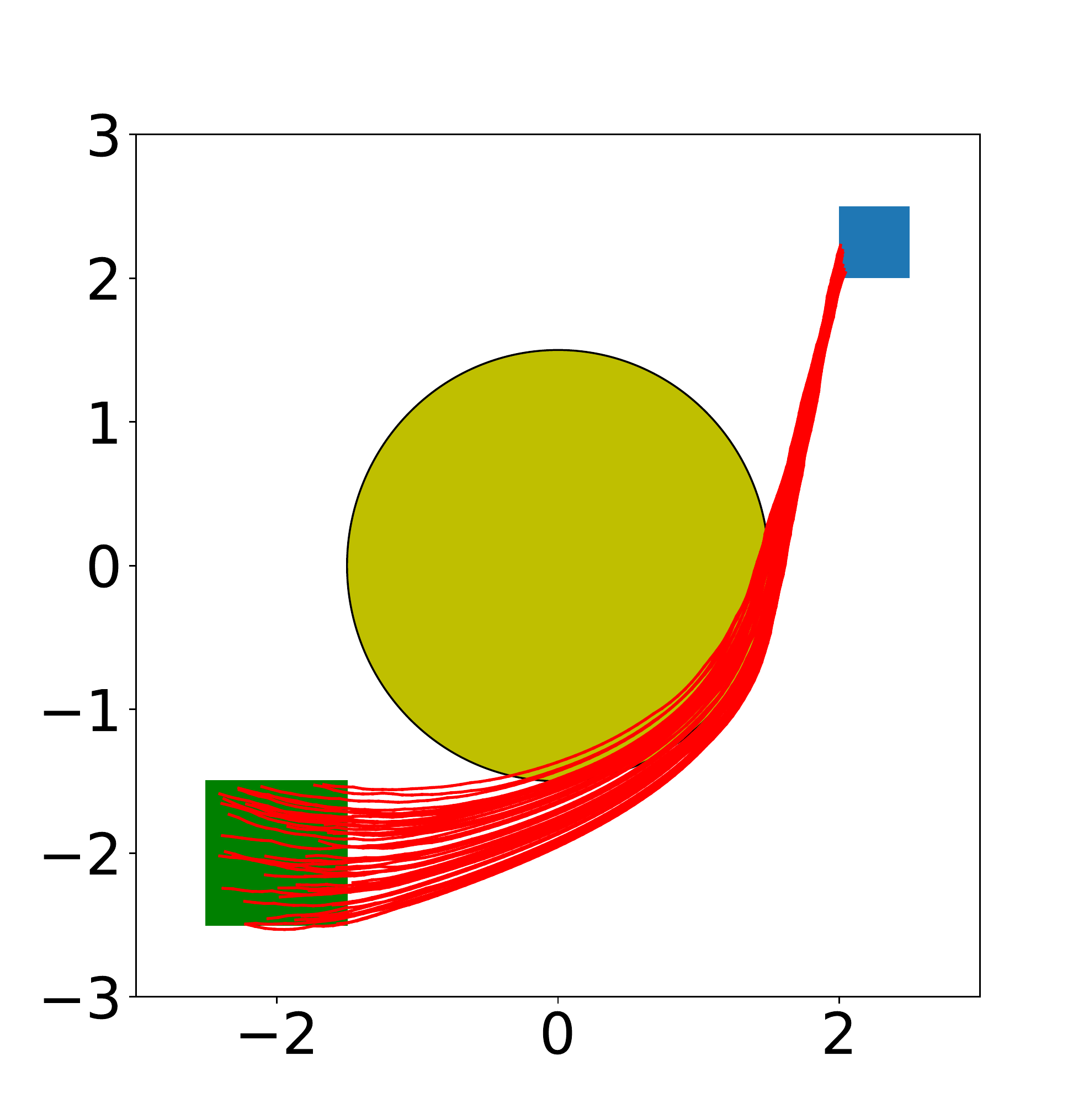}
	\caption{}
	\label{fig:5a}
	\end{subfigure}%
	\begin{subfigure}[h]{0.48\linewidth}
		\includegraphics[width=1\linewidth]{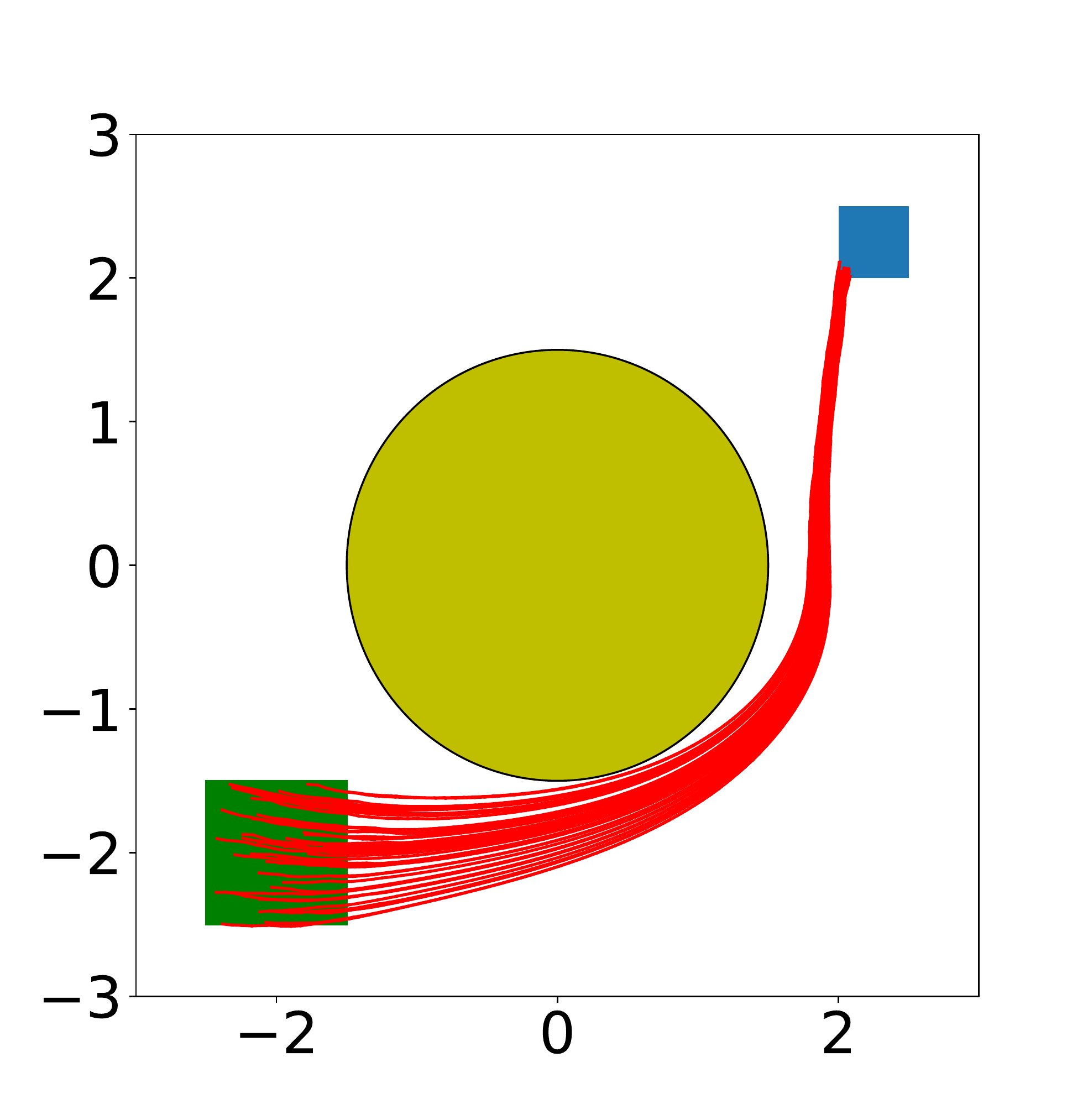}
	\caption{}
	\label{fig:5b}
	\end{subfigure}
	\caption{Safe rate comparison between using nominal CBF and the learned CBF with 50 trajectories. All the initial points are sampled from green areas (a): 50 trajectories using the nominal CBF.   (b): 50 trajectories using the learned CBF.}
\end{figure}

\begin{table}[htbp]
\centering
\begin{tabular}{lcccl}
\toprule
& Nominal CBF & Learned CBF & \\ 
\midrule
Number of samples & 50 & 50 &\\ 
Number of unsafe trajectories& 36 & 0 & \\
Safe rate & $28\%$ & $100\%$ &\\ 
\bottomrule
\end{tabular}
\caption{Safe rate between the nominal CBF and the learned CBF for Experiment 1.}
\label{Tab:2}
\end{table}

\subsection{Experiment 2}
In the second experiment, we test our algorithm in a more complicated working scenario with multiple obstacles. The robot is traveling in the working space, in which there are two static people at $(-2,1)$ and $(2,1)$ with a safe radius of $0.5$ and $1$, respectively. There is also a pedestrian on a path in the working space that the robot should not run into. For simplicity, we do not consider the dynamics of the pedestrian, but only enclose the potential positions using an ellipse. As a result, the safety criteria for the robot can be interpreted as not entering the pink regions as in Fig~\ref{fig:6a} and Fig~\ref{fig:6b}. We use one control barrier function for each pink region so that we have three CBFs $h_1=(x+2)^2+(y-1)^2-0.5^2$, $h_2=(x-1)^2+(y-1)^2-1$ and $h_3=(x+1)^2+4(y+1)^2-1$. We construct one neural network for each control barrier function. Each neural network has 2 hidden layers with 200 nodes in each layer. We test uncertainty on $u$ where the nominal model has $u=1$ and the real model $u=0.7$. The neural networks is trained using 40 sample points. The simulation result shows that the robot will run into the pink regions for the real model without learning the CBF and our algorithm can provide a safe trajectory for system with uncertainty using the trained neural networks. 

\begin{figure}[htbp]
	\centering
	\begin{subfigure}[h]{0.5\linewidth}
		\includegraphics[width=1\linewidth]{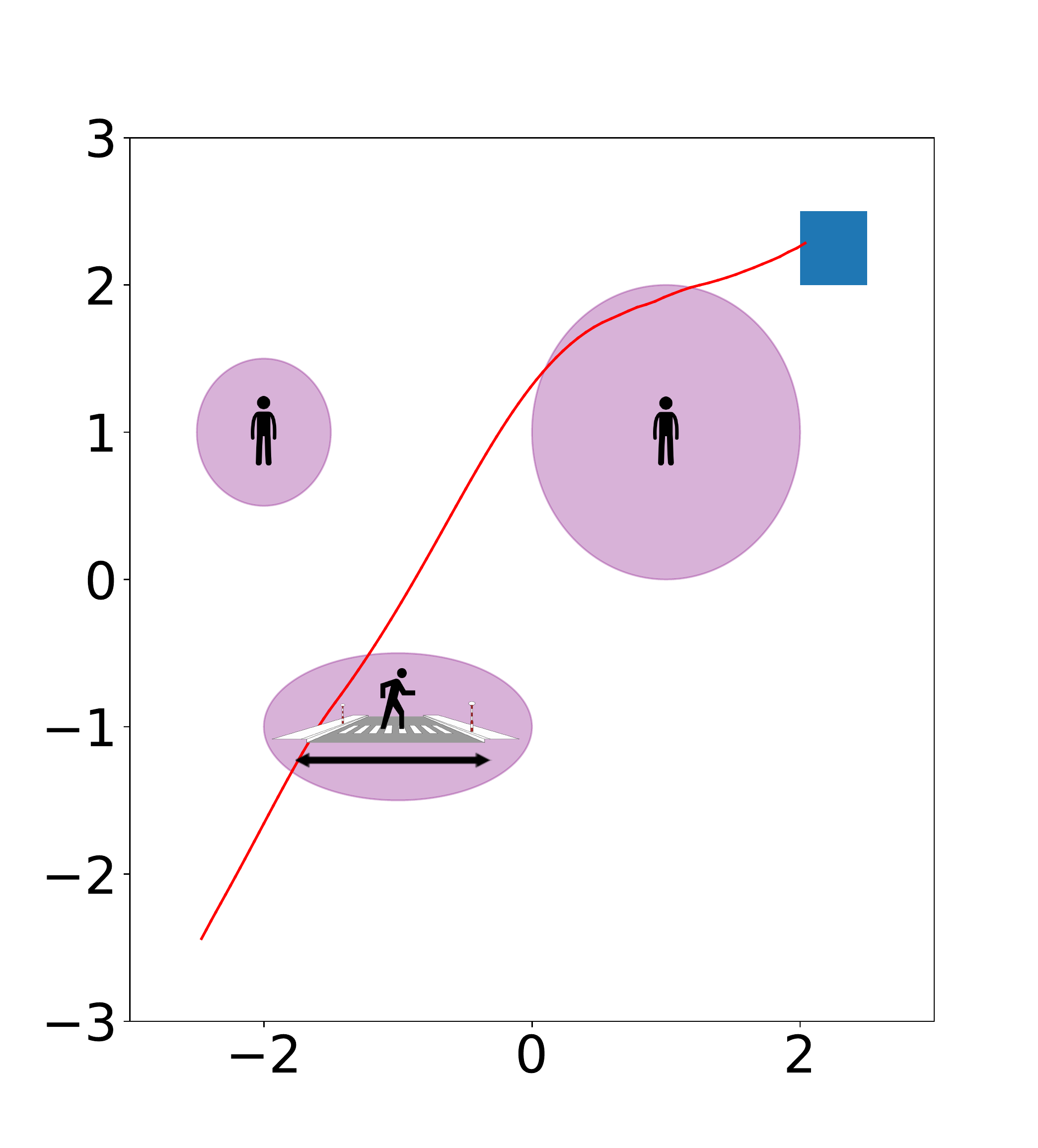}
	\caption{}
	\label{fig:6a}
	\end{subfigure}
	\begin{subfigure}[h]{0.48\linewidth}
		\includegraphics[width=1\linewidth]{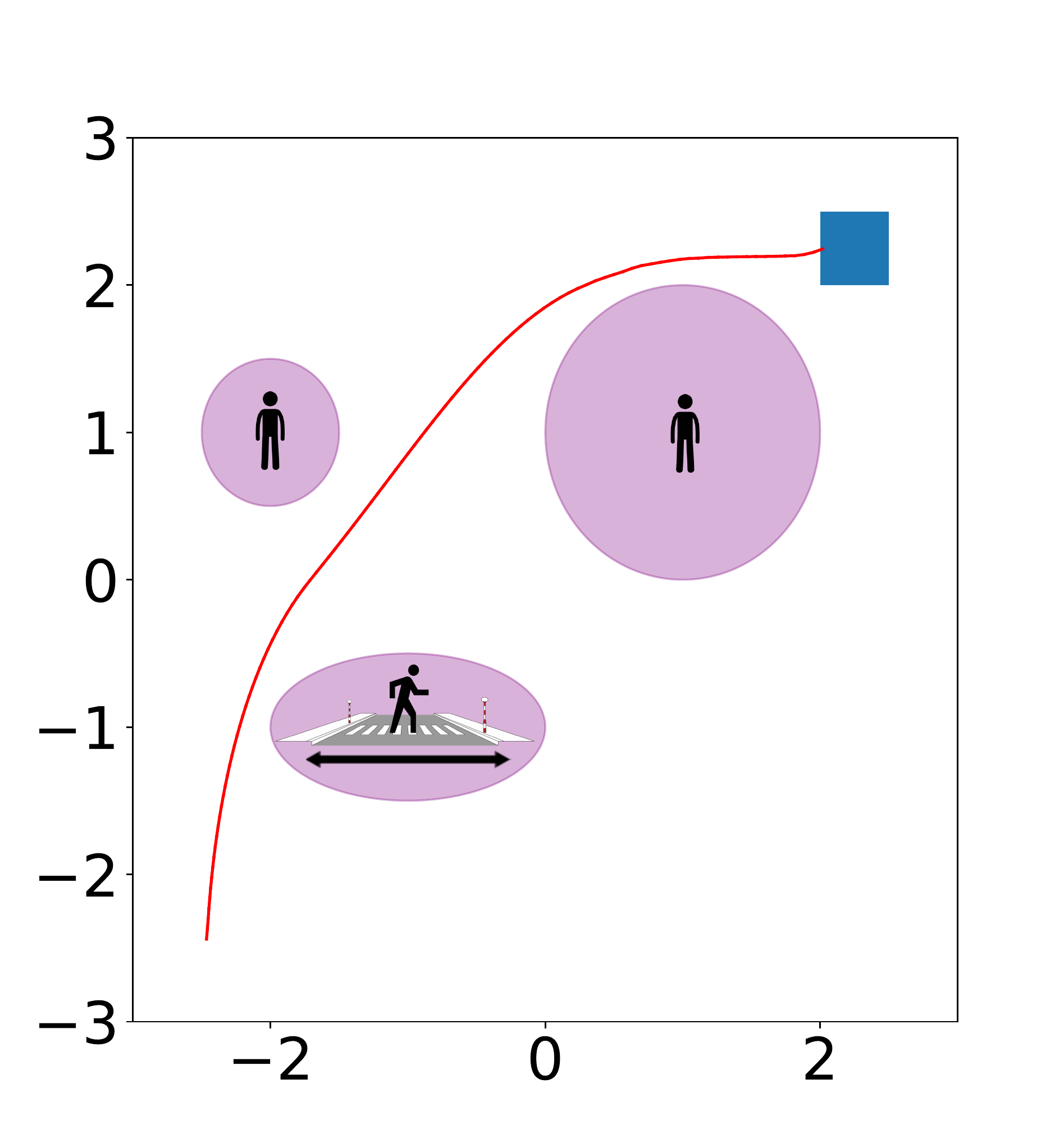}
	\caption{}
	\label{fig:6b}
	\end{subfigure}
	\caption{Simulation result of Experiment 2: the pink regions are areas that the robot should avoid. The blue square is the goal region. The initial state for the robot is (-2.5,-2.5,$\frac{\pi}{3}$). (a): The trajectory for the real model without a learned CBF. (b): The trajectory using a learned CBF.}
\end{figure}

\subsection{Experiment 3}
We test our algorithm using a dynamic obstacle in the third experiment. As is shown in Figure~\ref{fig:exp3_1}, the initial position of the robot is marked as the blue star. A moving obstacle moves along the x-axis to the right from $(-2,0)$ with a speed of $0.6/s$. The radius of the obstacle is $r_O=0.5$ and the goal is marked as the blue square. We use the control barrier function 
\begin{equation*}
    h=(x-x_O)^2+(y-y_O)^2-r_O^2,
\end{equation*}
where $x_O$ and $y_O$ are $x$ and $y$ coordinate of the obstacle. The parameters for the nominal model are the same as in the first experiment and we use the same nominal controller as well. The real system has the uncertainty that $u=0.7$. We also use the same structure of the neural network as in Experiment 1 and sample 40 trajectories for training. The trajectory using the learned CBF is shown in Figure~\ref{fig:exp3_1}. The yellow circle and star are the position of the obstacle and the robot at time step $n=50$ and the green circle and star are those for time step $n=70$. We also plot the value of $h$ during the simulation in Figure~\ref{fig:exp_h} for using the learned CBF and the nominal CBF. We see that the value of $h$ is always positive using the learned CBF while the value drops below $0$ for using the nominal CBF. This implies that the robot avoids the moving obstacle successfully when we use the learned CBF to solve the QP problems while it collides with the obstacle when we use the nominal CBF. Besides, we can see that the blue curve terminates earlier in Figure~\ref{fig:exp_h} than the red curve. This is because we terminate plotting $h$ when the robot reach the goal region. We also test the safe rate of our method. We sample 50 initial points to test the result and compare the safe rate using the nominal CBF and the learned CBF. As in Experiment 1, all the initial points are sample within in the region $[-2.5,-1.5]\times[-2.5,-1.5]$. The result is presented in TABLE~\ref{Tab:3}. From the table, we can see that our method guarantees a $100\%$ safe rate while using the nominal CBF for the real system, the success rate is only $36\%$.
\begin{figure}[htbp]
	\centering
	\begin{subfigure}[h]{0.5\linewidth}
		\includegraphics[width=1\linewidth]{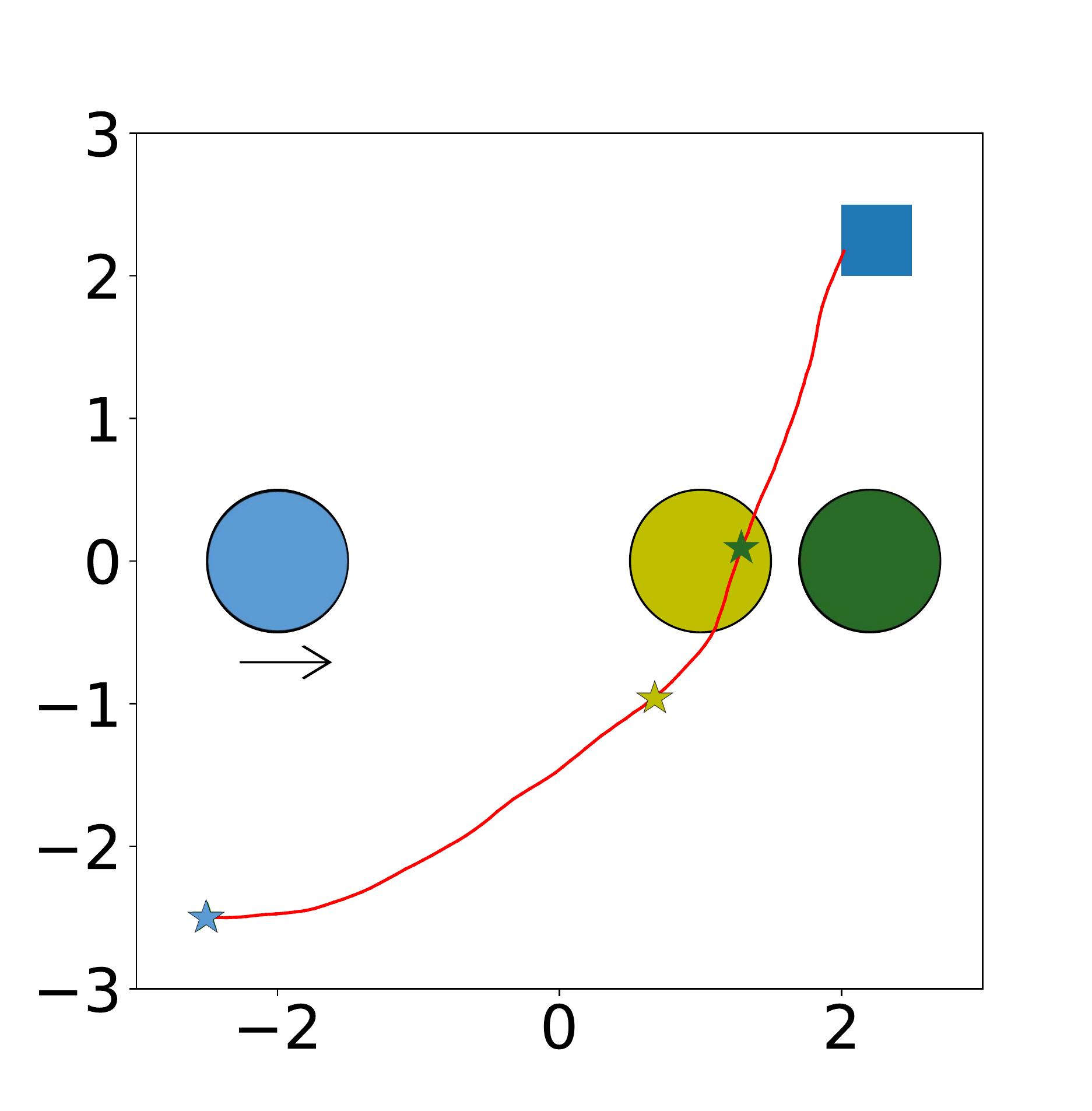}
	\caption{}
	\label{fig:exp3_1}
	\end{subfigure}%
	\begin{subfigure}[h]{0.5\linewidth}
		\includegraphics[width=1\linewidth]{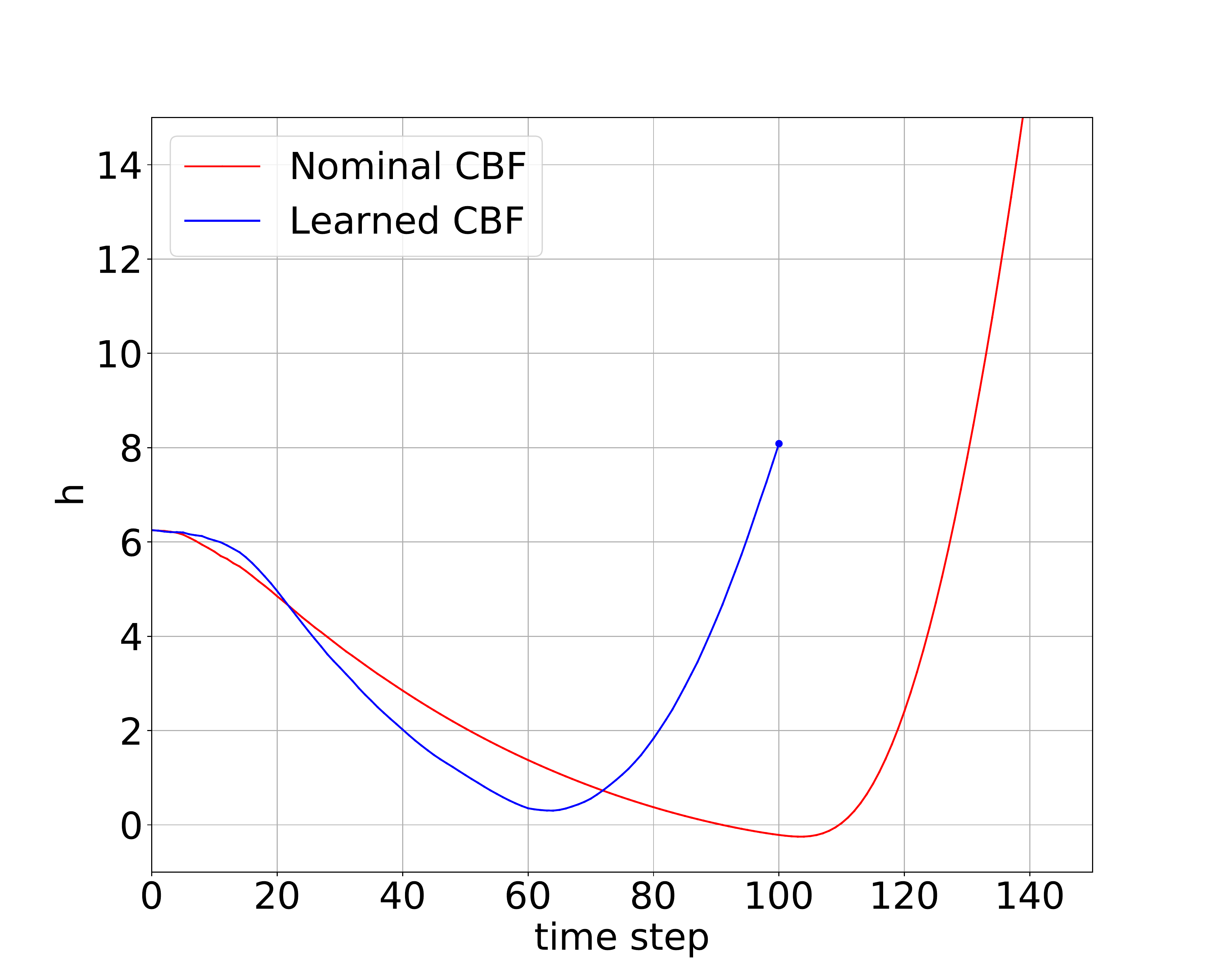}
	\caption{}
	\label{fig:exp_h}
	\end{subfigure}
	\caption{Simulation result for Experiment 3: (a): The initial position is at $(-2.5,-2.5,0)$ marked as the blue star. The obstacle is marked with the blue circle at $(-2,0)$ and moves right with a speed of $0.6/s$. The yellow star and circle are the positions of the robot and obstacle at time step $n=50$. The green star and circle are the positions of the robot and obstacle at time step $n=70$. The blue square is the goal region. The trajectory calculated using the learned CBF is the red curve. (b): The value of $h$ during simulation.}
\end{figure}

\begin{table}[htbp]
\centering
\begin{tabular}{lcccl}
\toprule
& Nominal CBF & Learned CBF & \\ 
\midrule
Number of samples & 50 & 50 &\\ 
Number of unsafe trajectories& 32 & 0 & \\
Safe rate & $36\%$ & $100\%$ &\\ 
\bottomrule
\end{tabular}
\caption{Safe rate between the nominal CBF and the learned CBF for Experiment 3.}
\label{Tab:3}
\end{table}
    
\section{Conclusion}
In this paper, we present a framework for learning the CBFs with high relative degree for systems with uncertainty. We first provide sufficient conditions on controllers via CBFs with high relative degree for set invariance. We also show that the dynamics of the real CBF can be learned from that of the nominal CBF and a remainder by using neural networks. We show in simulation that our method can handle model uncertainty using a differential driving robot model. Since we need to calculate high order derivative of the control barrier functions during training using numerical differentiation, the error in high order derivative will affect the performance of the networks. As a result, we will study the impact of numerical differentiation for the learning process. 

\bibliography{root}{}
\bibliographystyle{plain}
\end{document}